\documentclass{amsart}
\usepackage[cp1251]{inputenc}
\usepackage[english,russian]{babel}
\usepackage{amsmath}
\usepackage{amssymb}
\usepackage{amsfonts}
\def\mscs{65N20,35J05}
\newtheorem{lem}{Lemma}
\newtheorem{thm}{Theorem}
\begin{document}
MSCS \mscs
\thispagestyle{empty}

\title[The Cauchy problem and Hadamard's example in the ring
]{The Cauchy problem and Hadamard's example in the ring}

\author{Yaremko O.E.}

\address{Oleg Yaremko,
\newline\hphantom{iii}Penza State University
\newline\hphantom{iii} Lermontova st., 37,
\newline\hphantom{iii} 440038, Penza, Russia}
\email{yaremki@mail.ru}

\maketitle {\small

\begin{quote}
\noindent{\sc Abstract. }
Integral representation for harmonic function  in the ring is constructed in this work. \\
 We prove the existence and uniqueness of solutions of the Cauchy problem for the Laplace equation in the ring. Integral representation for the solution of the Cauchy problem was found.
\medskip

 \end{quote}

\section{introduction}
We now examine Hadamard's example of an ill-posed problem \cite{5,6}.
Consider Cauchy problem for the Laplace's equation
\[
\Delta u\equiv \frac{1}{r}\frac{\partial }{\partial r}\left(
{r\frac{\partial u}{\partial r}} \right)+\frac{1}{r^2}\frac{\partial
^2u}{\partial \varphi ^2}=0
\]
\[
\left. {u} \right|_{r=1} =0,
\left. {\frac{\partial  {u}}{\partial n}} \right|_{r=1} =e^{-\sqrt{n}}\frac{cosn\varphi}{n},
\]
$n$ is an add integer. Then
\[u(r,\varphi)=e^{-\sqrt{n}}\frac{r^{n}-r^{-n}}{2n^{2}} cos(n\varphi)
\]
 is the unique solution Cauchy problem. But now consider large $n$. We have \[||e^{-\sqrt{n}}\frac{cosn\varphi}{n}||\rightarrow0
   \]
for any Sobolev norm \cite{3} because the   beats all polynomials in n (from any finite order of derivatives). Hence we can choose $n$ so large that the boundary condition is as close to 0 as we please. But the solution $u(r,\varphi)$ grows large at any fixed line $r =constant$ as $n\rightarrow\infty$ ;
\[
||u(r,\varphi)-u_{0}||\rightarrow\infty
\]
hence $u_{0}\equiv0$ solves
\[
\Delta u\equiv \frac{1}{r}\frac{\partial }{\partial r}\left(
{r\frac{\partial u}{\partial r}} \right)+\frac{1}{r^2}\frac{\partial
^2u}{\partial \varphi ^2}=0
\]
\[
\left. {u} \right|_{r=1} =0,
\left. {\frac{\partial {u}}{\partial n}} \right|_{r=1} =0,
\]
uniquely. So this problem does not depend continuously on the data, and hence is not well-posed. Therefore it is not something we can safely ignore, well- posedness.

In this paper the analytical solution of the Hadamard problem \cite{1} in a ring is
resulted. The received results can be applied for solving inverse boundary
problems, type of the Hadamard problem.
Let us arrive at the function $u\left( {r,\varphi } \right)$ such that $u$ is
the solution of the Laplace's equation
\begin{equation}
\label{eq1}
\Delta u\equiv \frac{1}{r}\frac{\partial }{\partial r}\left(
{r\frac{\partial u}{\partial r}} \right)+\frac{1}{r^2}\frac{\partial
^2u}{\partial \varphi ^2}=0
\end{equation}
in the ring $K$: $r<\left| z \right|<1;z=re^{i\varphi}$, $u(r,\varphi$ is continuous in the closed ring
$\bar {K}$. Moreover, for the function $u$ the conditions hold
\begin{equation}
\label{eq2}
\left. u \right|_{r=1} =g(\varphi),
\quad
\left. {\frac{\partial u}{\partial n}} \right|_{r=1} =h\left( \varphi
\right),
\end{equation}
$h\left( \varphi \right)$ is continuous, $2\pi $- periodic function.
Let us show that the Hadamard problem solution is
\[
u\left( {r,\varphi } \right)=\frac{1}{2\pi }\int_0^\infty {e^{-\varepsilon
}} \int\limits_0^{2\pi } {Re\left[ {\exp \left( {\varepsilon ze^{-it}}
\right)} \right]} g\left( t \right)dt+
\]
\[
+\frac{1}{2\pi }\int_0^\infty {e^{-\varepsilon }} \int\limits_0^{2\pi }
{Re\left[ {\frac{1}{z}\exp \left( {\varepsilon \frac{e^{it}}{z}} \right)}
\right]} g\left( t \right)dt-\frac{1}{2\pi }\int\limits_0^{2\pi } {h\left(
\psi \right)\,d\psi } \cdot \ln r+
\]
\begin{equation}
\label{eq3}
+\frac{1}{2\pi }\int\limits_0^\infty {\frac{1}{\lambda }\left(
{e^{-\frac{\lambda \,}{r}}-e^{-\lambda r}} \right)\int\limits_0^{2\pi }
{e^{\lambda \cos \left( {\varphi -\psi } \right)}\cos \left( {\lambda \sin
\left( {\varphi -\psi } \right)} \right)h\left( \psi \right)\,d\psi }
\,d\lambda } .
\end{equation}

\section{the main lemma}
\begin{lem}If the function $w=f\left( z \right)$ is analytic in the
ring $r<\left| z \right|<1$, and continuous in the closed ring $r\le
\left| z \right|\le1$, then
\[
f\left( z \right)=\frac{1}{2\pi i}\int_0^\infty {e^{-\varepsilon }}
\int\limits_{C_{r_1 } } {e^{\varepsilon \frac{z}{\varsigma }}} f\left(
\varsigma \right)\frac{d\varsigma }{\varsigma }+
\]
\begin{equation}
\label{eq4}
+\frac{1}{2\pi i}\int_0^\infty {e^{-\varepsilon }} \int\limits_{C_{r_2 } }
{\frac{1}{z}} e^{\varepsilon \frac{\varsigma }{z}}f\left( \varsigma
\right)d\varsigma ,r<\left| z \right|<1.
\end{equation}
\end{lem}
\begin{proof}. Consider the sequence of function
\[
f_N \left( z \right)=\frac{1}{2\pi i}\int_0^N {e^{-\varepsilon }}
\int\limits_{C_{r_1 } } {e^{\varepsilon \frac{z}{\varsigma }}} f\left(
\varsigma \right)\frac{d\varsigma }{\varsigma }+
\]
\[
+\frac{1}{2\pi i}\int_0^N {e^{-\varepsilon }} \int\limits_{C_{r_2 } }
{\frac{1}{z}} e^{\varepsilon \frac{\varsigma }{z}}f\left( \varsigma
\right)d\varsigma ,r<\left| z \right|<1.
\]
From the Cauchy formula \cite{4} it follows that the contour of integration can be
transferred to a circles $C_1 ,C_r $. We get
\[
f_N \left( z \right)=\frac{1}{2\pi i}\int_0^N {e^{-\varepsilon }}
\int\limits_{C_1 } {e^{\varepsilon \frac{z}{\varsigma }}} f\left( \varsigma
\right)\frac{d\varsigma }{\varsigma }+
\]
\[
+\frac{1}{2\pi i}\int_0^N {e^{-\varepsilon }} \int\limits_{C_2 }
{\frac{1}{z}} e^{\varepsilon \frac{\varsigma }{z}}f\left( \varsigma
\right)d\varsigma ,r<\left| z \right|<1.
\]
Interchanging the order of integration in each summand and integrating
the~inner integrals with respect to $\varepsilon $, we obtain
\[
-\frac{1}{2\pi i}\int\limits_{C_r } {\left[ {\frac{1}{\varsigma
-z}-\frac{e^{-N\left( {1-\frac{\varsigma }{z}} \right)}}{\varsigma -z}}
\right]} f\left( \varsigma \right)d\varsigma ,r<\left| z \right|<1.
\]
We pass to the limit as $N\to \infty $ and get the following estimates:
\[
\left| {\frac{1}{2\pi i}\int\limits_{C_1 } {\left[ {\frac{e^{-N\left(
{1-\frac{z}{\varsigma }} \right)}}{\varsigma -z}} \right]} f\left( \varsigma
\right)d\varsigma } \right|\le \frac{1}{2\pi }\int\limits_{C_1 }
{\frac{e^{-N\left( {1-\left| z \right|} \right)}}{1-\left| z \right|}}
\left| {f\left( \varsigma \right)} \right|\left| {d\varsigma } \right|\le
\frac{e^{-N\left( {1-\left| z \right|} \right)}}{1-\left| z \right|}M,
\]
\[
\left| {\frac{1}{2\pi i}\int\limits_{C_r } {\left[ {\frac{e^{-N\left(
{1-\frac{\varsigma }{z}} \right)}}{\varsigma -z}} \right]} f\left( \varsigma
\right)d\varsigma } \right|\le \frac{1}{2\pi }\int\limits_{C_r }
{\frac{e^{-N\left( {1-\frac{r}{\left| z \right|}} \right)}}{\left| z
\right|-r}} \left| {f\left( \varsigma \right)} \right|\left| {d\varsigma }
\right|\le \frac{e^{-N\left( {1-\frac{r}{\left| z \right|}} \right)}}{\left|
z \right|-r}M,
\]
\[M=\mathop {\max }\limits_{r\le \left| z \right|\le 1} \left| {f\left( z
\right)} \right|.
\]

 From the estimates it follows that the summands

\[
\frac{1}{2\pi i}\int\limits_{C_1 } {\left[ {\frac{e^{-N\left(
{1-\frac{z}{\varsigma }} \right)}}{\varsigma -z}} \right]} f\left( \varsigma
\right)d\varsigma ,
\frac{1}{2\pi i}\int\limits_{C_r } {\frac{e^{-N\left( {1-\frac{\varsigma
}{z}} \right)}}{\varsigma -z}} f\left( \varsigma \right)d\varsigma ,r<\left|
z \right|<1,
\]
tend to 0. This is uniform convergence in $z$ in any of the closed ring
$r_1 \le \left| z \right|\le r_2 $. Thus we have $\mathop {\lim
}\limits_{N\to \infty } f_N \left( z \right)=f\left( z \right).$
\end{proof}
\section{main results}

Consider now the Hadamard problem \cite{1}.
\begin{thm}
Let us arrive at the function $u\left( {r,\varphi } \right)$ such that $u$ is
the solution of the Laplace's equation
\begin{equation}
\label{eq5}
\Delta u\equiv \frac{1}{r}\frac{\partial }{\partial r}\left(
{r\frac{\partial u}{\partial r}} \right)+\frac{1}{r^2}\frac{\partial
^2u}{\partial \varphi ^2}=0
\end{equation}
in the ring $K$:$r<\left| z \right|<1$, $u$  is continuous in the closed ring
$\bar {K}$. Moreover, for the function $u$ the conditions hold
\begin{equation}
\label{eq6}
\left. u \right|_{r=1} =g\left( \varphi \right),
\quad
\left. {\frac{\partial u}{\partial n}} \right|_{r=1} =0,
\end{equation}
$g\left( \varphi \right)$ is continuous, $2\pi $- periodic function,
the solution of the
problem (\ref{eq5})-(\ref{eq6}) is expressed in the form :
\[
u\left( {r,\varphi } \right)=\frac{1}{2\pi }\int_0^\infty {e^{-\varepsilon
}} \int\limits_0^{2\pi } {Re\left[ {\exp \left( {\varepsilon ze^{-it}}
\right)} \right]} g\left( t \right)dt+
\]
\[
+\frac{1}{2\pi }\int_0^\infty {e^{-\varepsilon }} \int\limits_0^{2\pi }
{Re\left[ {\frac{1}{z}\exp \left( {\varepsilon \frac{e^{it}}{z}} \right)}
\right]} g\left( t \right)dt,r<\left| z \right|<1,
\]
\[
z=re^{i\varphi }.
\]
\end{thm}
\textbf{Proof}. By the Cauchy -- Riemann equations \cite{2}, so that $-iv_\varphi ^/ =ru_r^/ $.We
get $v\left( {1,\varphi } \right)=0$. It follows that we get $\left. f
\right|_{r=1} =g\left( \varphi \right)$ for the function $f=u+iv$.

Substituting $r_1 =r_2 =1$ in (\ref{eq4}), we obtain
\[
f\left( z \right)=\frac{1}{2\pi }\int_0^\infty {e^{-\varepsilon }}
\int\limits_0^{2\pi } {\exp \left( {\varepsilon ze^{-it}} \right)} g\left( t
\right)dt+
\]
\[
+\frac{1}{2\pi }\int_0^\infty {e^{-\varepsilon }} \int\limits_0^{2\pi }
{\frac{1}{z}} \exp \left( {\varepsilon \frac{e^{it}}{z}} \right)g\left( t
\right)dt,r<\left| z \right|<1.
\]
We take a real part of $f\left( z \right)$. Theorem is proved.

\begin{thm}
Let us arrive at the function $u\left( {r,\varphi } \right)$ such that $u$ is
the solution of the Laplace's equation
\begin{equation}
\label{eq5}
\Delta u\equiv \frac{1}{r}\frac{\partial }{\partial r}\left(
{r\frac{\partial u}{\partial r}} \right)+\frac{1}{r^2}\frac{\partial
^2u}{\partial \varphi ^2}=0
\end{equation}
in the ring $K$:$r<\left| z \right|<1$,$ u$ is continuous in the closed ring
$\bar {K}$. Moreover, for the function $u$ the conditions hold
\begin{equation}
\label{eq6}
\left. u \right|_{r=1} =0,
\quad
\left. {\frac{\partial u}{\partial n}} \right|_{r=1} =h(\varphi),
\end{equation}
$g\left( \varphi \right)$ is continuous, $2\pi $- periodic function,
The solution of the
problem (\ref{eq5})-(\ref{eq6}) is expressed in the form:
\[
u\left( {r,\varphi } \right)=-\frac{1}{2\pi }\int\limits_0^{2\pi } {h\left(
\psi \right)\,d\psi } \cdot \ln r+
\]
\[
+\frac{1}{2\pi }\int\limits_0^\infty {\frac{1}{\lambda }\left(
{e^{-\frac{\lambda \,}{r}}-e^{-\lambda r}} \right)\int\limits_0^{2\pi }
{e^{\lambda \cos \left( {\varphi -\psi } \right)}\cos \left( {\lambda \sin
\left( {\varphi -\psi } \right)} \right)h\left( \psi \right)\,d\psi }
\,d\lambda } .
\]
\[
z=re^{i\varphi }.
\]
\end{thm}

\textbf{Proof}. Consider the
Cauchy auxiliary problem:
\[
\Delta \tilde {u}\equiv \frac{1}{r}\frac{\partial }{\partial r}\left(
{r\frac{\partial \tilde {u}}{\partial r}}
\right)+\frac{1}{r^2}\frac{\partial ^2\tilde {u}}{\partial \varphi ^2}=0,
\]
\[
\left. {\tilde {u}} \right|_{r=1} =g(\varphi),
\left. {\frac{\partial \tilde {u}}{\partial n}} \right|_{r=1} =0,
\]
where
\[
g( \varphi )=\frac{1}{2\pi }\int\limits_0^{2\pi } {\psi
h\left( {\varphi +\psi } \right)d\psi } .
\]

The solution of the problem (\ref{eq5})-(\ref{eq6}) is by theorem 2:
\[
u\left( {r,\varphi } \right)=\frac{1}{2\pi }\int_0^\infty {e^{-\varepsilon
}} \int\limits_0^{2\pi } {Im\left[ {\exp \left( {\varepsilon ze^{-it}}
\right)} \right]} g(t)dt+
\]
\[
+\frac{1}{2\pi }\int_0^\infty {e^{-\varepsilon }} \int\limits_0^{2\pi }
{Im\left[ {\frac{1}{z}\exp \left( {\varepsilon \frac{e^{it}}{z}} \right)}
\right]} g( t )dt,r<\left| z \right|<1.
\]
Integrating by parts, and using
\[
g'( \varphi)=h\left( \varphi \right),
\]
\[
\widetilde{u}( r,\varphi )=\frac{1}{2\pi }\int\limits_0^{2\pi } {\psi
u\left(r, {\varphi +\psi } \right)d\psi },
\]
we get (\ref{eq3}).

\end{document}